\documentclass[sigconf]{acmart}

\usepackage{booktabs} 
\usepackage{todonotes} 

\usepackage{tikz}

\setcopyright{rightsretained}

\copyrightyear{2017}
\acmYear{2017}
\setcopyright{acmlicensed}
\acmConference[DIDL'17]{DIDL'17: Workshop on Distributed Infrastructures for Deep Learning}{December 11--15, 2017}{Las Vegas, NV, USA}
\acmPrice{15.00}
\acmDOI{10.1145/3154842.3154843}
\acmISBN{978-1-4503-5169-0/17/12}

\begin{document}
\title{The TensorFlow Partitioning and Scheduling Problem: \\It's the Critical Path!}

\author{Ruben Mayer, Christian Mayer, Larissa Laich}
\affiliation{
	\institution{\{firstname.lastname\}@ipvs.uni-stuttgart.de}
	\institution{Institute of Parallel and Distributed Systems}
	\city{University of Stuttgart, Germany}
}

\renewcommand{\shortauthors}{Ruben Mayer et al.}

\renewcommand{\shorttitle}{The TensorFlow Partitioning and Scheduling Problem: It's the Critical Path!}

\begin{abstract}
State-of-the-art data flow systems such as TensorFlow impose iterative calculations on large graphs that need to be partitioned on heterogeneous devices such as CPUs, GPUs, and TPUs. However, partitioning can not be viewed in isolation. Each device has to select the next graph vertex to be executed, i.e., perform local scheduling decisions. Both problems, partitioning and scheduling, are NP-complete by themselves but have to be solved in combination in order to minimize overall execution time of an iteration. In this paper, we propose several heuristic strategies to solve the partitioning and scheduling problem in TensorFlow. We simulate the performance of the proposed strategies in heterogeneous environments with communication-intensive workloads that are common to TensorFlow. Our findings indicate that the best partitioning and scheduling heuristics are those that focus on minimizing the execution time of the critical path in the graph. Those strategies provide a speed-up of up to 4 times in comparison to strategies that are agnostic to the critical path, such as hash-based partitioning and FIFO scheduling.
\end{abstract}


%
%
%



\maketitle

\begin{tikzpicture}
\begin{scope}[overlay]
\node[text width=20cm] at ([yshift=-7.0cm]current page.south) {(c) Owner 2017. This is the authors' version of the work. It is posted here for your personal use. Not for redistribution. \newline The definitive version is published in Proceedings of DIDL '17: Workshop on Distributed Infrastructure for Deep Learning, https://doi.org/10.1145/3154842.3154843 .};
\end{scope}
\end{tikzpicture}

\newcommand{\device}{dev} 
\newcommand{\devices}{D} 
\newcommand{\numDevices}{k} 
\newcommand{\capacity}{C} 
\newcommand{\power}{s} 
\newcommand{\bandwidthMatrix}{B} 

\newcommand{\graph}{G} 
\newcommand{\vertices}{V} 
\newcommand{\vertex}{v} 
\newcommand{\numVertices}{n} 
\newcommand{\edges}{E} 
\newcommand{\edge}{e} 
\newcommand{\numEdges}{m} 
\newcommand{\comp}{c} 
\newcommand{\comm}{t} 

\newcommand{\commScore}{c_s} 

\newcommand{\problem}{\textit{TF}} 

\newcommand{\collocationConstraints}{\mathbb{C}} 
\newcommand{\deviceConstraints}{\mathbb{D}} 

\newcommand{\scheduleFunc}{f} 
\newcommand{\partFunc}{p} 

\section{Introduction}
\label{sec:introduction}

General-purpose distributed data processing systems for graph-structured data have experienced a phenomenal growth in the last decade. Because of their superior performance and scalability, highly-optimized systems such as Google's TensorFlow \cite{Abadi2016TensorFlow} have created the foundation of the recent breakthroughs in the field of machine learning - focusing on deep learning applications such as image and video classification, natural language processing, and world-class Go playing machines.

Distributed machine learning systems such as TensorFlow express the computation as a directed data flow graph where graph vertices represent computational operations and edges transport data between these operations.
This  abstraction empowers the data scientist to harness the power of multi-core infrastructures while expressing arbitrary complex computations.
The major objective is to minimize execution time of the given computational task. Deep learning applications often require the graph to be executed in an iterated fashion with many iterations of supervised learning algorithms on huge data sets. Hence, graphs can be executed thousands of times. For instance, training neural networks can take tens of hours on dozens of devices \cite{dean2012large}. 
Especially the choice of partitioning the graph onto multiple devices and scheduling graph vertices for execution on the devices has huge impact on the performance of the data flow system in terms of execution time. Both problems, i.e., partitioning and scheduling, are NP-complete and, hence, are only feasible if solved in a heuristic fashion.

However, although the problem of partitioning and scheduling is a key challenge to minimize execution time of the data flow graphs, we identified significant lack of research in the following areas. First, there is no formal description of the TensorFlow partitioning and scheduling problem and, as a consequence, no proof about the NP-completeness. Second, existing partitioning heuristics focus on minimizing the amount of communication rather the execution time of the data flow graph, while existing scheduling approaches assume global scheduling decisions only rather than pushing scheduling logic to the devices. 

The goal of this paper is to conduct a thorough study on a wide spectrum of partitioning and scheduling algorithms as a basis to decide which algorithms to use for distributed deep learning problems. 
In particular, our contributions are as follows. (1) We formulate the TensorFlow partitioning and scheduling problem (denoted as \problem{}) and prove NP-hardness by reducing from the well-known NP-complete single execution time scheduling \cite{Ullman1975NP}. (2) We develop several partitioning and scheduling heuristics specifically tailored to \problem{} that are based on intuitive ideas. (3)  We evaluate the performance of the proposed heuristics on different graphs in a simulation. Results indicate that  the best partitioning and scheduling heuristics are based on minimizing the execution time of the critical path, providing a speed-up of up to 4 times in comparison to strategies that are agnostic to the critical path such as hash-based partitioning and FIFO scheduling.

\section{Problem Formulation}
\label{sec:problemFormulation}

\begin{figure}
	\includegraphics[width=0.45\textwidth]{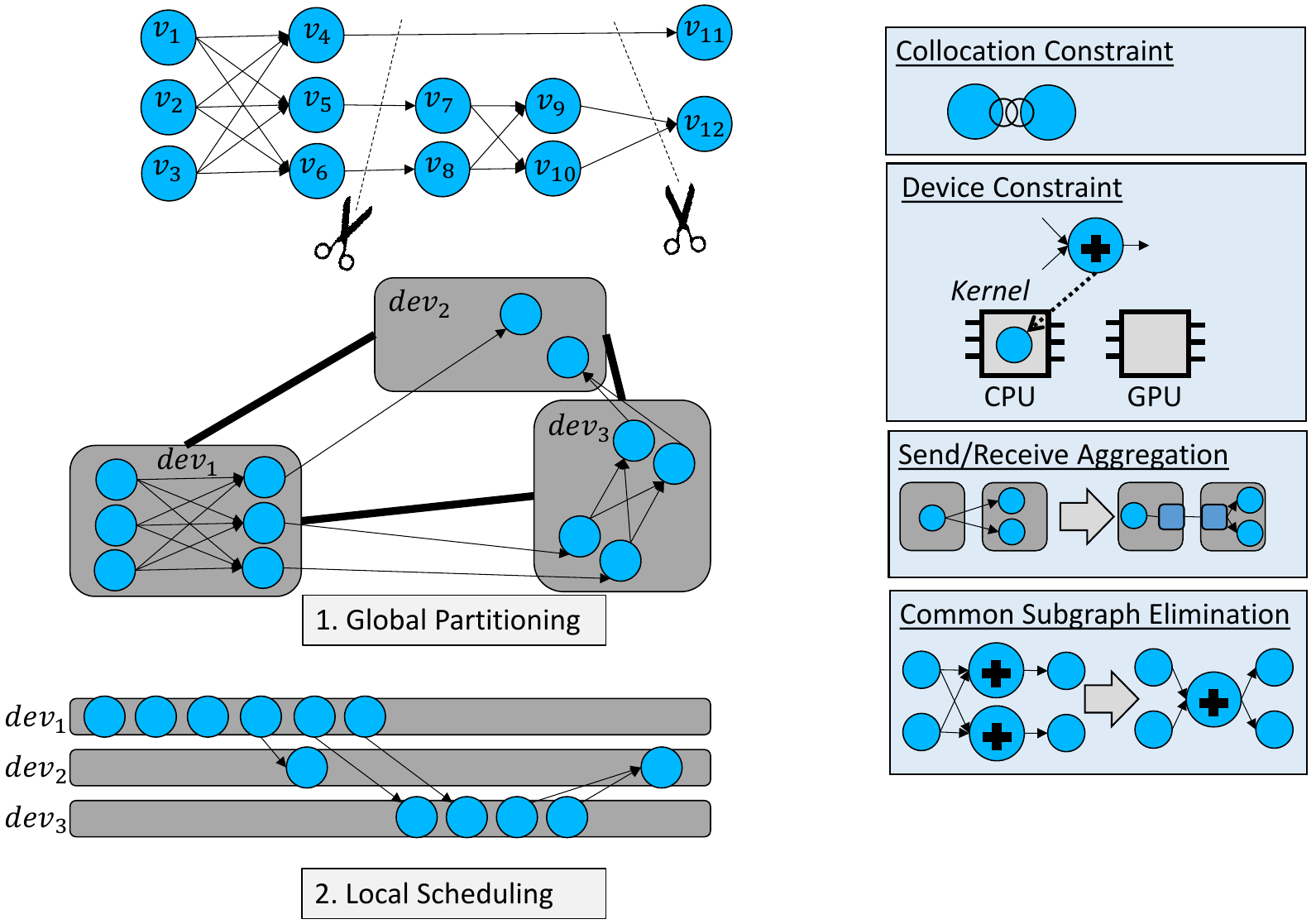}
	\vspace{-5pt}
	\caption{Problem Formulation.}
	\vspace{-10pt}
	\label{fig:problemFormulation}
\end{figure}

In this section, we formulate the TensorFlow partitioning and scheduling problem (abbreviated \problem{}).
 
Let graph $\graph=(\vertices, \edges)$ be the directed, acyclic data flow graph with vertices $\vertices=\{\vertex_1,...,\vertex_{\numVertices}\}$ and edges $\edges=\{\edge_1,...,\edge_{\numEdges}\} \in \vertices \times \vertices$.
Edges transport data from \textit{source to target vertices}, i.e., the output of the source vertex serves as input to the target vertex.
Associated to each edge $\edge_i \in \edges$ is the amount of communication $\comm_i \in \mathbb{R}$. For instance, in TensorFlow an edge is denoted as tensor, i.e., an array of (primitive) data values.
Vertices are denoted as \textit{schedulable} if data on all incoming edges is available. 
Associated to each vertex $\vertex_i \in \vertices$ is its computational complexity $\comp_i \in \mathbb{R}$.

Let $\devices$ be set of devices $\devices{} = \{\device{}_1,...,\device{}_{\numDevices{}}\}$.
Each device $\device_i \in \devices$ has computational speed $\power_i \in \mathbb{R}$.
For example if vertex $\vertex_1$ with computational complexity $\comp_1=10$ (operations) is executed on device $\device_1 \in \devices$ with computational speed $\power_1=10$ (operations per second), then the execution takes $\frac{\comp_1}{\power_1}=1$ second.
Furthermore, device $\device_1 \in \devices$ has maximal memory capacity $\capacity_i \in \mathbb{R}$. For instance, data leaving the source vertex of an edge consumes memory of the device -- that can have more or less severe memory restrictions.
Each two devices are connected via a (physical or virtual) network link.
The bandwidth is given by the bandwidth matrix $\bandwidthMatrix \in \mathbb{R}^{\numDevices \times \numDevices}$.
For instance, devices $\device_1$ and $\device_2$ communicate with bandwidth $\bandwidthMatrix_{1,2}$ bytes per second.

The set of \textit{collocation constraints} $\collocationConstraints \in \vertices \times \vertices$ encodes the symmetric relation of vertices that have to be placed on the same device.
For instance, stateful operations need to be placed on the same device as their states \cite{Abadi2016TensorFlow}.
Moreover, real-world data flow applications come with implicit or explicit \textit{placement constraints} of computational operations 
denoted as device constraints $\deviceConstraints \in \vertices \times \devices$. 
Examples are restrictions of the kernel, i.e., the concrete implementation of an operation, or of the hardware, e.g. GPU preferences.

In Figure~\ref{fig:problemFormulation}, we give an example.
The directed acyclic graph (DAG) consists of twelve vertices that are partitioned onto three devices. This partitioning is global, i.e., a logically centralized algorithm selects vertices to be assigned to devices. After the partitioning, each device contains a set of vertices to be executed. There may be several schedulable vertices at each point in time, e.g., device $\device_1$ could schedule vertices $\vertex_1$, $\vertex_2$, or $\vertex_3$.

The partitioning function $\partFunc: \vertices \to \devices$ assigns vertices to devices and the scheduling function $\scheduleFunc: \vertices \to \mathbb{N}$ assigns vertices to time slots in which the vertices are executed.
The goal is to minimize the execution time of the global schedule: 

\vspace{-0.2cm}
\begin{equation}
min_{\scheduleFunc}( max_{\vertex \in \vertices} \scheduleFunc(\vertex)),
\end{equation}

Note that the function $\scheduleFunc$ returns the \textit{starting time} of a vertex execution while we are interested in minimizing the maximal \textit{finishing time}. But we can overcome this problem easily by connecting all vertices in the DAG with out-degree zero (i.e., sink vertices) via a zero cost graph edge to an artificial final sink vertex with computation complexity zero. The starting time of the artificial sink vertex relates to the finishing time of overall computation.

We have to ensure that the memory constraint is fulfilled (cf. Equation~\ref{eq:memoryConstraint}), i.e., in each point in time $l \in \mathbb{N}$ and for each device $\device_j \in \devices$ the total memory usage for keeping data on all input edges of not yet scheduled vertices does not exceed the maximal capacity $\capacity_j$. We denote the set of active edges on device $\device_j$ at time $l$ as $\edges_{\mathit{active}}(l,j)$.

\vspace{-0.2cm}
\begin{equation}
\label{eq:memoryConstraint}
	\forall \device_j \in \devices, l \in \mathbb{N}: \sum_{\edge_i \in \edges_{\mathit{active}}(l,j)} \comm_i < \capacity_j
\end{equation}

Additionally, we require the collocation constraints to hold, i.e., 

\vspace{-0.2cm}
\begin{equation}
\label{eq:collocotationConstraint}
	\forall \vertex_i, \vertex_j \in \vertices: (\vertex_i, \vertex_j) \in \collocationConstraints \rightarrow \partFunc(\vertex_i)=\partFunc(\vertex_j).
\end{equation}

Finally, the device constraints have to hold, i.e.,
\begin{equation}
\label{eq:deviceConstraint}
\forall \vertex_i \in \vertices, \device_j \in \devices: (\vertex_i, \device_j) \in \deviceConstraints \rightarrow \partFunc(\vertex_i)=\device_j.
\end{equation}

\begin{figure*}
	\includegraphics[width=0.6\textwidth]{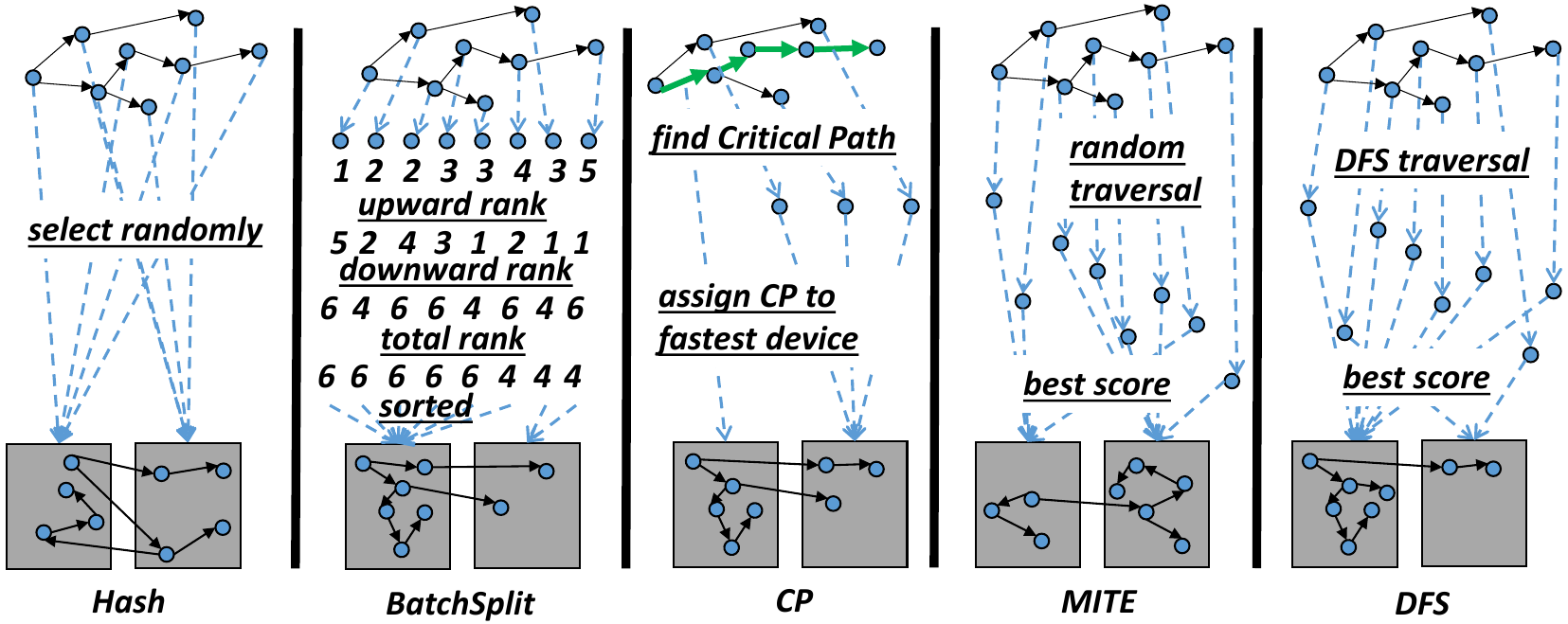}
	\vspace{-10pt}
	\caption{Overview: partitioning strategies.}
	\vspace{-10pt}
	\label{fig:partitioningApproaches}
\end{figure*}

\subsection{NP-completeness}
\begin{theorem}
	\problem{} is NP-complete
\end{theorem}

\begin{proof}
	
	First, we reduce the NP-complete single execution time scheduling \cite{Ullman1975NP}, denoted as \textit{Scheduling}, to Decision-\problem{}, i.e., \\$max_{\vertex \in \vertices} (\scheduleFunc(\vertex))<t_{max}$ (Lemma~\ref{lemma1}).
	This is the associated decision problem to \problem{} and therefore is in the same complexity class \cite{Ullman1975NP}. 
	Second, we show that Decision-\problem{} is in NP (Lemma~\ref{lemma2}).
	
	\begin{lemma}
		\label{lemma1}
		\textit{Scheduling} can be reduced to Decision-\problem{} in polynomial time.
	\end{lemma}

	\begin{proof}
		The NP-complete single execution time scheduling is the following problem (cf. \cite{Ullman1975NP}):
		Given a set $S$ of jobs that take unit time on any device, a partial order $\prec$ on the set of jobs, $\numDevices'$ processors, and a time limit $t_{max}$. Is there a scheduling function $g:S \to \{0,...,t_{max}-1\}$ such that the following three properties hold? (i) The scheduling function respects the ordering relation, i.e., $v \in S \prec v' \in S \rightarrow g(v) < g(v')$. (ii) The time limit is not exceeded, i.e., $\forall v \in S: g(v)<t_{max}$. (iii) There are at most $\numDevices'$ active jobs at each point in time, i.e., $\forall i \in \{0,...,t_{max}-1\}: |\{v \in S | g(v)=i\}| \leq \numDevices'$.
	
		We reduce \textit{Scheduling} to Decision-\problem{} using the following method. The graph $\graph=(\vertices, \edges)$ is given by the set of vertices $\vertices=S$ and the set of edges $\edges$, where there is an edge $(\vertex_1,\vertex_2)$ iff the ordering relation defines $\vertex_1 \prec \vertex_2$. We set the number of devices $\numDevices=\numDevices'$. Furthermore, we set $\forall i,j \in \{0,...,|\vertices|\}$ to $\comm_i=\bandwidthMatrix_{i,j}=0$, i.e., we ignore network communication and data to be transported over edges. We set the computational complexity $\forall i,j \in \{0,...,|\vertices|\}$ to $\comp_i=1$ and the device speeds $\forall i \in \{0,...,\numDevices\}$ to $\power_i=1$, such that processing each vertex on each device takes unit time. There are no collocation and device constraints, i.e., $\collocationConstraints=\emptyset$ and $\deviceConstraints=\emptyset$. Clearly, setting these parameters to constants is a polynomial time reduction.
	
		Next, we show: Decision-$\mathit{\problem{}}==\mathit{true} \leftrightarrow  \mathit{Scheduling}==\mathit{true}$.
		\begin{itemize}
			\item[$\rightarrow$] Let Decision-$\mathit{\problem{}}==\mathit{true}$. This directly implies that there exists a scheduling function $\scheduleFunc$ that satisfies the maximal execution time, i.e., $min_{\scheduleFunc}( max_{\vertex \in \vertices} \scheduleFunc(\vertex))<t_{max}$.
			The schedule given by the function $\scheduleFunc$ satisfies the three conditions of the \textit{Scheduling} problem: (i) the ordering is respected (the DAG directly encodes the required partial ordering $\prec$ as constructed by the reduction), (ii) the time limit $t_{max}$ is respected as $max_{\vertex \in \vertices} \scheduleFunc(\vertex)<t_{max} \rightarrow \forall v \in S: f(v)<t_{max}$, and (iii) there are at most $k$ active jobs (there are only $k$ devices). Hence, \textit{Scheduling}$==true$.
			\item[$\leftarrow$] Let $\mathit{Scheduling}==\mathit{true}$, i.e., the properties (i), (ii), and (iii) are fulfilled for a scheduling function $g$. Because of (ii) the time limit is kept, i.e., $\forall v \in S: g(v)<t_{max}$, and therefore $min_{\scheduleFunc}( max_{\vertex \in \vertices} \scheduleFunc(\vertex))<t_{max}$. Hence, Decision-\problem{}$==true$.
		\end{itemize}
	\end{proof}
	
	\begin{lemma}
		\label{lemma2}
		Decision-\problem{} is in NP.
	\end{lemma}

	\begin{proof}
		In order to show that a problem is in NP, we have to find an algorithm that decides whether a given problem instance solves Decision-\problem{}. This is achieved by the following two steps.
		First, we calculate the maximal execution time $max_{\vertex \in \vertices} (\scheduleFunc(\vertex))$ and check whether it is smaller than $t_{max}$. Second, we determine whether the memory, collocation, and device constraints are fulfilled. Clearly, these validation methods can be performed in polynomial time.
	\end{proof}

	As Decision-\problem{} is in NP and the NP-complete \textit{Scheduling} problem can be reduced to Decision-\problem{}, Decision-\problem{} and hence, \problem{}, are NP-complete.
\end{proof}

\subsection{Challenges}
Several specific challenges have to be addressed for TensorFlow Partitioning and Scheduling:

\textbf{Scalability: } Existing scheduling algorithms globally select the processor and the start time for each vertex \cite{heft}. We argue that this might be feasible for task scheduling, where the tasks are relatively coarse-grained computational units. However, for large numbers of fine-grained graph vertices the time needed to calculate the global schedule might be too high. Additionally, the potential for stale schedules increases dramatically if millions of small errors of estimating computational complexity accumulate. Therefore, we propose to first partition the graph using scalable partitioning heuristics and then solve the scheduling problem locally on the machines. This is the de facto standard for TensorFlow partitioning.

\textbf{Heterogeneity: } In TensorFlow, there are heterogeneities on every level: the computational complexity of vertices, the communication volume of edges, the memory capacity of devices, the computational speed of devices, and the bandwidth between devices. All of these values can be highly heterogeneous in real-world deployments as shown in \cite{Abadi2016TensorFlow}.

\section{Partitioning Approaches}
\label{sec:partitioning}
In this section, we describe strategies to partition the graph such that the local scheduling algorithms running on the devices can exploit the locality and idle time is minimal. An overview of the proposed strategies is depicted in Figure \ref{fig:partitioningApproaches}.

\subsection{Hashing}
The simplest strategy is to randomly assign vertices to devices proportionally to the capacity of the devices by using a hash function. Collocation constraints are considered by iteratively merging collocated vertices into a collocation group and assigning this collocation group randomly. Device constraints are fulfilled by restricting the potential set of devices for each vertex to those devices to which assignment is allowed and choosing a random device out of those. Memory requirements are fulfilled in a similar manner by excluding devices with insufficient spare capacity.
Random assignment is very simple and fast and requires minimal space and time complexity leading to good scalability. However, hashing ignores locality of the graph vertices, leading to high communication overhead and slow execution of the DAG.

\subsection{Path-Based Heuristics}
\label{sec:pathHeu}
Instead of randomly assigning vertices to devices, it might be better to place the vertices that are on the \emph{critical path} all on the fastest device(s). The critical path is the path in the dataflow graph that has the longest computation time from source to sink vertex. Speeding up the processing of the critical path, hence, would speed up the overall computation time of the dataflow graph. 

We propose two different strategies that try to optimize the assignment of the critical path. The first strategy, \emph{Batch Split}, estimates the critical path by means of calculating the \emph{rank} of each vertex and then assigns batches of vertices that have the highest ranks to the fastest devices. The second strategy, \emph{Critical Path}, is also based on calculating the vertex ranks, but tries to assign the complete critical path to the fastest device. 

In particular, we define two ranks for each vertex, the \emph{upward rank} and the \emph{downward rank}. 
The upward rank of vertex $\vertex_i$ is  the summed computational complexity of vertices along the longest path from $\vertex_i$ to any sink vertex and defined as: 

\vspace{-0.4cm}
\begin{equation}
\label{eq:upwardRank}
\mathit{upRank}(v_i) = \mathtt{max}_{v_j \in \mathit{succ}(v_i)}(\mathit{upRank}(v_j) + c_i)
\end{equation}

The downward rank of $\vertex_i$ is the summed computational complexity of predecessor vertices along the longest path from any source vertex to $\vertex_i$ and defined as: 

\vspace{-0.5cm}
\begin{equation}
\label{eq:downwardRank}
\mathit{downRank}(v_i) = \mathtt{max}_{v_j \in \mathit{pred}(v_i)}(\mathit{downRank}(v_j) + c_i)
\end{equation}

\subsubsection{Batch Split}

The idea of the Batch Split strategy is to avoid expensive critical path computation but still prioritize optimal placement of vertices that are located \emph{on long paths}.

First, we calculate the \emph{total rank} for each vertex $\vertex$ that is the sum of the upward rank  and the downward rank. The total rank can be calculated efficiently by traversing the whole graph similar to Dijkstra's algorithm and calculate upward and downward ranks \textit{on the way}. 
Then, we sort vertices by total rank in descending order and assign them independently to the fastest feasible device. 
Due to the sorting of vertices, the complexity of the Batch Split strategy is in $\mathcal{O}(log(\numVertices) \times \numVertices)$ (the calculation of total rank can be performed in linear runtime).

%
%

\subsubsection{Critical Path}
The Critical Path (CP) strategy tries to achie-ve minimal execution time by assigning the complete critical path to the fastest device. With this strategy, no additional communication latency is added to the computation time of the critical path.

Hence, we first compute the critical path based on downward ranks, which works as follows. (1) The algorithm starts at the source vertices and computes for each vertex of the complete graph the downward rank. (2) Choose the sink vertex with the maximum downward rank. This vertex is on the critical path. (3) Follow the predecessor relation from the chosen sink vertex to the source, adding the visited vertices to the critical path. In case a vertex has multiple predecessors, follow all such paths. (4) When reaching a source vertex, a path is completed. Choose the longest path between the chosen sink vertex and any of the connected sink vertices in the predecessor relation of the sink vertex. This is the critical path.

The critical path is assigned to the fastest feasible device. If no device can hold the complete critical path, it is divided among the fastest devices.  
A vertex $\vertex_k$ that is not on the critical path is assigned to the device $\device_i$ with minimal sum of execution times of already assigned vertices plus the execution time of $\vertex_k$  on $\device_i$ (cf. Equation \ref{eq:cp}).

\vspace{-0.2cm}
\begin{equation}
\label{eq:cp}
\min_{\mathit{\device_i \in \devices}} \bigg( \Big( \sum_{\vertex_j \in V : p(\vertex_j) = \mathit{\device_i}} \frac{c_j}{\power_i} \Big) + \frac{c_k}{\power_i} \bigg)
\end{equation}


\subsection{Multi-Objective Heuristics}

Instead of only taking into account the critical path, we identify four objectives in assigning vertex $\vertex$ to device $\device$: (i) introduce minimal communication overhead, (ii) prefer fast devices, (iii) prefer devices with high memory capacity, and (iv) prefer vertices on critical paths to be assigned first. We propose two different heuristics that take into account that broader set of optimization objectives. 

\subsubsection{MITE}
\label{sec:mite}

The idea of the MITE (Memory, Importance, Traffic, and Execution time) strategy is to consider all four optimization objectives and use a heuristic function to assign the vertices to devices. A vertex $\vertex_i$ is assigned to the device $\device_l$ such that the following score is \emph{minimized}:

\vspace{-0.4cm}
\begin{equation}
\label{eq:mite}
\begin{split}
\mathit{mite}(\vertex_i, \device_l) = & \mathit{mem}(\device_l)\ \times \\
& \mathit{imp}(\vertex_i, \device_l) \times \\
& \mathit{traffic}(\vertex_i, \device_l)\ \times \\
& \mathit{execTime}(\vertex_i, \device_l)
\end{split}
\end{equation}

The memory score is the percentage of memory utilization of $\device_l$. Based on how many vertices are already assigned to $\device_l$ by the partitioning algorithm, the memory utilization will grow and less loaded devices will be favored by the partitioning algorithm.

The importance score favors important vertices to be placed on fast devices. An important vertex is a vertex with a high total rank (total rank is defined in Section \ref{sec:pathHeu} as the sum of upward and downward rank of a vertex). In the score computation, the total rank of a vertex is normalized by the highest total rank among all vertices, and the speed of a device is normalized by the speed of the fastest device. Formally, the importance score is defined as:

\vspace{-0.2cm}
\begin{equation}
\label{eq:importance}
\mathit{imp}(\vertex_i, \device_l) = 1 - (\frac{\mathit{totalRank}(\vertex_i)}{\mathtt{max}_{\vertex_j \in V}\{\mathit{totalRank}(\vertex_j)\}}\ \times\ \frac{\power_l}{\mathtt{max}_{\device_k \in D}\{ \power_k \}})
\end{equation}

The traffic score is computed as follows. If a vertex $\vertex$ is placed on device $\device$ and a neighboring vertex $\vertex' \neq \vertex$ is already assigned to another device $\device' \neq \device$, additional traffic is introduced. We define the traffic score as follows:

\vspace{-0.2cm}
\begin{equation}
\label{eq:communication}
\mathit{traffic}(\vertex_i, \device_l) = \sum_{\vertex_j \in \vertices, (\vertex_j,\vertex_i) \in \edges} \frac{\comm_i}{\bandwidthMatrix_{p(\vertex_j),l}}
\end{equation}

The execution time score is the predicted execution time of $\vertex_i$ on $\device_l$ as defined in Section \ref{sec:problemFormulation}, normalized by the maximum execution time of  $\vertex_i$ on any device. 

\subsubsection{Depth First Search}

The idea of Depth First Search (DFS) partitioning is to traverse the graph with a DFS algorithm and assign visited vertices using a multi-objective heuristic function. In particular, the DFS algorithm starts at the source vertex with the highest total rank (as defined in Section \ref{sec:pathHeu}). Whenever a vertex is visited, it is assigned to the device that minimizes the function

\vspace{-0.2cm}
\begin{equation}
\label{eq:dfs}
\mathit{dfsScore}(\vertex_i, \device_l) = \mathit{traffic}(\vertex_i, \device_l)\ \times\ \mathit{execTime}(\vertex_i, \device_l)
\end{equation}

with traffic and execution time score as defined in the MITE heuristics in Section \ref{sec:mite}.

\section{Scheduling}
\label{sec:scheduling}

In this section, we describe the scheduling algorithm that runs on each device (device is a broad term that can also refer to, e.g., a CPU core). We assume that a partitioning algorithm has already performed the partitioning decision, such that to device, a part of the graph is assigned. Now, the scheduling algorithm decides the execution order of the graph vertices on that device.

In particular, the function $\scheduleFunc: \vertices \to \mathbb{N}$ assigns the device's vertices to numbers which define the order of execution. The following criteria are set by the scheduling problem:
(1) Each vertex is executed exactly once in a scheduling cycle. (2) A device can execute at most one vertex simultaneously.
(3) Scheduling of a vertex is non-preemptive, i.e., the running execution of a vertex on a device cannot be interrupted.
(4) A vertex $v$ can only be executed when all tensors of the predecessors in the graph are computed and transferred to the device of $v$. When all ingoing tensors are available, the vertex is denoted \emph{executable}.
(5) A device can only schedule vertices that are assigned to it by the partitioning algorithm. A vertex cannot be assigned to multiple devices.
(6) A device has to remain idle when there are no executable vertices assigned to it.

Besides the classical non-preemptive scheduling algorithms such as FIFO, we have devised two scheduling algorithms tailored to the TensorFlow problem. 

\emph{Highest Path Computation Time first} (PCT) scheduling prefers the execution of the vertex whose longest path of direct and indirect successors takes most computation time. The rationale behind this strategy is to minimize the execution time of the \emph{critical path} on the graph. 

\emph{Maximum Successor Rank first} (MSR) scheduling tries to minimize idleness across all devices, by taking into account how many downstream vertices \emph{on the global graph} depend on a vertex (i.e., the rank) and rewards vertices whose execution activate downstream devices that are currently idle. The rationale behind this strategy is to maximize the \emph{resource utilization} of the devices.

\subsection{PCT Scheduling}

PCT schedules always the executable vertex whose longest path of direct and indirect successors takes most computation time. The upwards path computation time (PCT) is calculated from the sink vertices as shown in Equation~\ref{eq:pct}. The computation time of a sink vertex $v$ is the execution time of $v$ on its assigned device.  If $v$ is not a sink vertex, the maximal transfer time of a tensor to any of the successor vertices as well as the path computation time of that successor vertex are added. The transfer time is considered 0 if both vertices are placed on the same device. Else, it is the edge weight (size of the tensor) divided by the the transfer rate between the two devices.

\vspace{-0.2cm}
\begin{equation}
\label{eq:pct}
\mathit{PCT(v_i)} = \mathtt{max}_{v_j \in \mathit{succ}(v_i)} \Big(\mathit{PCT(v_j)} + \mathit{trans(v_j, v_i)}\Big) + \frac{\mathit{c_i}}{\mathit{\power_{\partFunc(v_i)}}}
\end{equation}

PCT is calculated once for each vertex, after the partitioning is decided. The calculated schedule is used for every iteration of the TensorFlow program.

\subsection{MSR Scheduling}
\label{sec:MSR}

The drawback of PCT scheduling is that the interdependency of vertices on different devices are not taken into account. However, it could be beneficial to schedule vertices on one device $a$ first that are necessary to activate vertices on another device $b$ in order to avoid $b$'s idleness. 

Our Maximum Successor Rank First (MSR) scheduling algorithm takes into account such dependencies. Each vertex $v$ on a device is scored not only based on the number of successor vertices, but also on the number of successor \emph{devices} (i.e., on how many devices the successor vertices of $v$ are distributed), and on the number of successor devices that are idle and get \emph{activated} when the vertex has been executed. The idea is to minimize idleness over all devices, such that throughput will be maximized.

The MSR scoring function is specified as follows:

\begin{equation}
\label{eq:msr}
\begin{split}
\mathit{MSR(\vertex_i)} = \sum_{\vertex_j \in \mathit{succ}(\vertex_i)} & \Big(\ \alpha * 1 \\
& + \beta * \big(\partFunc(\vertex_j) \neq \partFunc(\vertex_i)\big)  \\ 
& + \gamma * (\lvert \mathit{pred}(\vertex_j) \rvert = 1)\\
& + \delta * \big(\mathit{idle}(\partFunc(\vertex_j)) \land ( \lvert \mathit{pred(v_j)} \rvert = 1 ) \big)\ \Big)
\end{split}
\end{equation}

The weights $\alpha$, $\beta$, $\gamma$ and $\delta$ can be used to balance the scoring such that specific goals are emphasized. E.g., a relatively high value of $\delta$ favors the scheduling of vertices that lead to the activation of other idle devices.


\section{Evaluations}
\label{sec:evaluations}

To evaluate the proposed partitioning and scheduling strategies, we employ an event-based simulation. This allows us to evaluate large real-world networks on a large simulated number of devices.


\begin{table}
\centering
\caption{Properties of the evaluated neural networks.}
\vspace{-0.2cm}
\label{tab:syntheticSampleGraph}
\resizebox{0.45\textwidth}{!}{%
\begin{tabular}{|c|c|c|c|c|}
\hline
Graph & \#nodes & \begin{tabular}[c]{@{}c@{}}\#edges\end{tabular} & \begin{tabular}[c]{@{}c@{}}average\\ node degree\end{tabular} & \begin{tabular}[c]{@{}c@{}}\#colocated \\nodes\end{tabular} \\ \hline
convolutional\_network     & 347   & 531                      & 1.53      & 104                            \\ \hline
recurrent\_network     & 3,069   & 5,533            & 1.8        & 533                             \\ \hline
dynamic\_rnn     & 5,271    & 9,214                     & 1.75                      & 1,356                     \\ \hline
\end{tabular}%
}
\vspace{-0.2cm}
\end{table}

\begin{figure*}
	\includegraphics[width=0.56\textwidth]{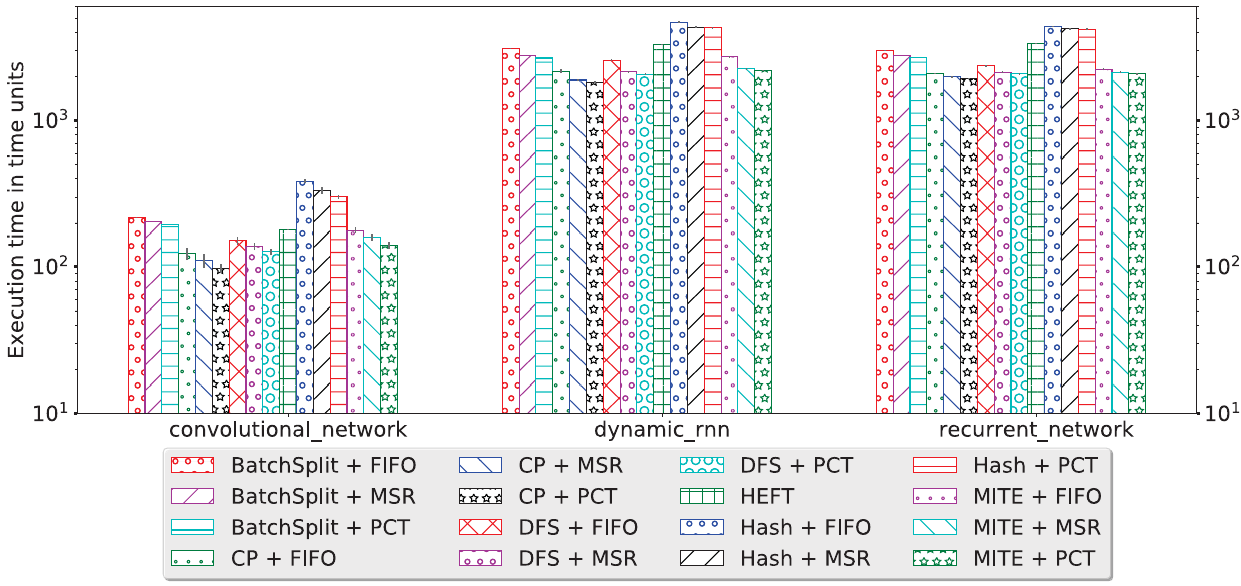}
	\vspace{-10pt}
	\caption{Execution time of 1 iteration at different partitioning and scheduling strategies on 50 devices.}
	\vspace{-10pt}
	\label{fig:evalExecPart}
\end{figure*}

\subsection{Experimental setup}

\textbf{Simulation Parameters.}
 In the evaluations, we simulated 50 devices. The simulated devices had assigned a speed in a random range of 10 to 100 operations per time unit. The transfer rate between the devices is set to a random number in the range between 10 and 60 bytes per time unit.  Tensor sizes were again randomly set between 1 and 100 bytes. The number of operations to compute a vertex function was randomly set between 1 and 100 operations. The numbers are based on real tensor sizes and on devices available from Amazon AWS.

\textbf{TensorFlow Networks.}
Three real TensorFlow networks were extracted from the TensorFlow examples on Github\footnote{\url{https://github.com/aymericdamien/TensorFlow-Examples/tree/master/examples/3_NeuralNetworks}}.
The properties of these networks are listed in Table \ref{tab:syntheticSampleGraph}. The co-location constraints were directly taken from the TensorFlow networks.

\textbf{Baselines.}
The Heterogeneous Earliest Finish Time (HEFT) algorithm \cite{heft} solves the related \textit{task scheduling problem} but can not be used directly for the TensorFlow partitioning and scheduling problem.
We modified the HEFT algorithm slightly in the processor selection phase in order to ensure satisfaction of the TensorFlow constraints.
First, we consider only \textit{feasible devices} for a vertex $\vertex$, i.e., devices to which collocated vertices are already assigned and that satisfy the device constraints of vertex $\vertex$. 
In this phase the earliest finish time of a vertex for each \textbf{feasible} device is computed and assigned to the device with the lowest earliest finish time.
If the vertex is collocated, all collocated vertices are assigned but as their predecessor vertices might be not assigned yet the earliest finish time can not be computed. So their computing time slot is only computed and added to the devices time schedule when it is their turn due to the task prioritization list.

For scheduling, a second baseline we consider is FIFO scheduling, i.e., the vertices are scheduled according to the time they become executable. If two vertices become executable at the same time, the FIFO scheduler chooses randomly which of them to schedule.

\subsection{Experiments and Results}
Six different partitioning strategies were executed in combination with three different scheduling strategies on the three networks  on 50 simulated devices. For MSR scheduling, the weights were set to $\alpha = 1, \beta = 1, \gamma = 1, \delta = 5$. Some of the strategies are non-deterministic as, e.g., the order of vertices being assigned to devices might differ. 

The results of the simulation runs are depicted in Figure \ref{fig:evalExecPart}. The visualized execution time is the average of 10 executions and the standard mean deviation is shown as a gray line on each bar.  The results show that both partitioning and scheduling have a significant impact on the TensorFlow performance. The best strategy (CP partitioning + PCT scheduling) is up to 4 times as fast as the worst strategy (i.e., Hash partitioning + FIFO scheduling). This was the case throughout all tested TensorFlow networks.

The reason for the bad performance of Hash partitioning and FIFO scheduling is that those strategies are agnostic to the characteristics of the \problem{} problem. Hash partitioning leads to good load balancing, but this is not most important in reducing the data flow computation time from source to sink. FIFO scheduling does not favor the fast execution of the critical path either. It becomes clear that the focus of both partitioning and scheduling strategies on reducing the computation time of the critical path is of immense importance.

\section{Related Work}
\label{sec:relatedWork}


Graph partitioning has drawn a lot of attention as a preprocessing step in large-scale graph processing \cite{7536511, Verma:2017:ECP:3055540.3055543}. However, graph partitioning strategies for graph processing focus on minimizing the traffic between devices and keeping the load balanced. In contrast to that, the \problem{} problem deals with a data flow from source vertices to sink vertices. Hence, partitioning strategies that are successful for graph processing problems may not generalize to the \problem{} problem.

The HEFT \cite{heft} scheduling algorithm evaluated in this paper has shown good performance in comparison to 20 other heuristic scheduling strategies for data flow execution \cite{Canon2008}. However, such scheduling algorithms place vertices (or tasks) on devices (or processors) dynamically at run-time. We do not consider a relocation of vertices at run-time from one device to another as a promising way to tackle the \problem{} problem. Instead, the more intuitive approach is to partition the graph first and then perform the scheduling locally on the devices. 
\section{Conclusion}
\label{sec:conclusion}

In this paper, we have formally defined the TensorFlow partitioning and scheduling problem, proven the NP-completeness of the problem, and proposed a number of heuristics that solve the problem. The evaluation results based on a simulation indicate that those partitioning and scheduling strategies yield the best results that focus on the minimization of the execution time of the critical path, outperforming other strategies by up to a factor of 4. 

In the future work, we plan to implement the proposed strategies in the TensorFlow framework to verify the simulation results. In particular, we aim to extend the tested deep learning networks to other types and larger scales, to see whether the findings generalize. 

\bibliographystyle{ACM-Reference-Format}
\bibliography{bib/bib} 

\end{document}